\documentclass[copyright]{eptcs}
\providecommand{\event}{B. Klin, P. Soboci\'nski (Eds.): \\
6th Workshop on Structural Operational Semantics (SOS'09)}
\providecommand{\volume}{19}
\providecommand{\anno}{2010}
\providecommand{\firstpage}{62}
\providecommand{\eid}{5}
\usepackage{breakurl}
\usepackage{graphics,graphicx}
\usepackage{amssymb,amsmath}
\usepackage{amsthm}
\usepackage{url}
\usepackage{psfrag}

\usepackage{epsfig}
\usepackage{graphics,graphicx}
\theoremstyle{plain}
\newtheorem{thm}{Theorem}[section]
\newtheorem{theorem}[thm]{Theorem}
\newtheorem{cor}[thm]{Corollary}
\newtheorem{lemma}[thm]{Lemma}
\newtheorem{prop}[thm]{Proposition}

\theoremstyle{definition}
\newtheorem{defn}[thm]{Definition}
\newtheorem{example}[thm]{Example}

\theoremstyle{remark}
\newtheorem{rem}[thm]{Remark}
\newtheorem*{claim}{Claim}

\newcommand{\Defeq}{\stackrel{\mathrm{df}}{=}}
\newcommand{\union}{\cup}
\newcommand{\Union}{\bigcup}
\newcommand{\inter}{\cap}
\newcommand{\cross}{\times}
\newcommand{\pow}{\mc P}
\newcommand{\res}{\restriction}

\newcommand{\tranc}[1]{\stackrel{#1}{\rightarrow}_{\mc C}}
\newcommand{\trand}[1]{\stackrel{#1}{\rightarrow}_{\mc D}}
\newcommand{\dtranc}[2]{\stackrel{#1,#2}{\rightarrow_{\mc C}}}
\newcommand{\dtrand}[2]{\stackrel{#1,#2}{\rightarrow_{\mc D}}}
\newcommand{\rtranc}[1]{\stackrel{#1}{\rightsquigarrow}_{\mc C}}
\newcommand{\rtrand}[1]{\stackrel{#1}{\rightsquigarrow}_{\mc D}}
\newcommand{\rdtranc}[2]{\stackrel{#1,#2}{\rightsquigarrow_{\mc C}}}
\newcommand{\rdtrand}[2]{\stackrel{#1,#2}{\rightsquigarrow_{\mc D}}}
\newcommand{\card}[1]{|#1|}
\newcommand{\lab}{\ell}

\newcommand{\Par}{\mathrel{\mbox{$\! \mid\! $}}}

\newcommand{\Comment}[1]{}

\newcommand{\tran}[1]{\stackrel{#1}{\rightarrow}}

\newcommand{\rtran}[1]{\stackrel{#1}{\rightsquigarrow}}

\newcommand{\mc}[1]{\mathcal{#1}}

\newcommand{\mn}[1]{\mathsf{min}(#1)}
\newcommand{\depth}[2]{\mathsf{depth}_{#1}(#2)}
\newcommand{\Nat}{\mathbb{N}}

\newcommand{\Act}{\mathsf{Act}}
\newcommand{\ibeqt}{\approx_{ib}}
\newcommand{\sbeqt}{\approx_{sb}}
\newcommand{\hheqt}{\approx_{hh}}
\newcommand{\rbeqt}{\approx_{rb}}
\newcommand{\rsbeqt}{\approx_{rsb}}
\newcommand{\dbeqt}{\approx_{db}}
\newcommand{\rdbeqt}{\approx_{rdb}}
\newcommand{\rhsbeqt}{\approx_{rhsb}}
\newcommand{\rhesbeqt}{\approx_{rhesb}}
\newcommand{\co}{\mathrel{co}}
\newcommand{\Cstable}{\mathbb{C}_{stable}}

\title{Reverse Bisimulations on Stable Configuration Structures}
\author{Iain Phillips
\institute{Department of Computing, Imperial College London, England}
\email{iccp@doc.ic.ac.uk}
\and
Irek Ulidowski
\institute{Department of Computer Science, University of Leicester, England}
\email{iu3@mcs.le.ac.uk}
}
\def\titlerunning{Reverse Bisimulations on Stable Configuration Structures}
\def\authorrunning{I.C.C. Phillips \& I. Ulidowski}
\begin{document}
\maketitle

\begin{abstract}
The relationships between various equivalences on configuration structures, including interleaving bisimulation (IB), step bisimulation (SB) and hereditary history-preserving (HH) bisimulation, have been investigated by van Glabbeek and Goltz (and later Fecher).  Since HH bisimulation may be characterised by the use of reverse as well as forward transitions, it is of interest to investigate forms of IB and SB where both forward and reverse transitions are allowed.
We give various characterisations of reverse SB, showing that forward steps do not add extra power.  We strengthen Bednarczyk's result that, in the absence of auto-concurrency, reverse IB is as strong as HH bisimulation, by showing that we need only exclude auto-concurrent events at the same depth in the configuration.

\end{abstract}

\section{Introduction}
The relationships between various equivalences on configuration structures, including interleaving bisimulation (IB), step bisimulation (SB) and hereditary history-preserving (HH) bisimulation, have been investigated by van Glabbeek and Goltz~\cite{vGG01} (and later Fecher~\cite{Fec04}).
Since HH bisimulation may be characterised by the use of reverse as well as forward transitions, it is of interest to investigate forms of IB and SB where both forward and reverse steps are allowed.
We give various characterisations of RSB, showing among other things that forward steps do not add extra power.  We strengthen Bednarczyk's result that, in the absence of auto-concurrency, IB with reverse transitions (which we shall call reverse bisimulation or RB) is as strong as HH bisimulation, by showing that we need only exclude auto-concurrent events at the same depth in the configuration.

In this paper we adopt as our model of concurrency the framework of {\em stable configuration structures},
following van Glabbeek and Goltz.
These can be regarded as a more abstract version of stable event structures~\cite{Win87}.
A {\em configuration} is a set of events which is a possible run within an event structure.
Using stable configuration structures means that our results can be directly related to those of van Glabbeek and Goltz, and of Fecher.
In our earlier work on reversible computation~\cite{PU07a} we employed {\em prime} event structures,
which are a special case of stable ESs.
Prime event structures are technically simpler;
they have a global causal ordering on events, whereas with
stable event structures causal orderings are parametrised by configuration.
Stable event structures have the advantage over prime event structures that such operations as
sequential and parallel composition can be modelled more easily.
See e.g.~\cite{vGG01} for further discussion of the merits of the various forms of event structure.

The question then arises as to when two configuration structures should be regarded as equivalent.
We do not want to distinguish configuration structures merely on the basis that they have different sets of events.
So we assume that events come with a labelling (we use $a,b,c$ for labels); events with the same label are equivalent.
Simple configuration structures can be written in a CCS-style notation~\cite{Mil89}, with conflict represented by choice ($+$),
concurrency by parallel composition ($\; \Par\; $), and causal ordering by action prefixing ($.$).
As a simple example, the law $a+a = a$ will hold for any reasonable equivalence on configuration structures,
even though the configuration structure represented by $a+a$ has two (conflicting) events and $a$ only has one.

{\em Hereditary history-preserving} (HH) bisimulation is considered as the finest well-behaved true concurrency
equivalence relation that ``completely respects the causal and branching structure 
of concurrent systems and their interplay''~\cite{vGG01}.
Here ``well-behaved'' means that it preserves refinement of actions.
HH bisimulation was proposed by Bednarczyk~\cite{Bed91}.
It can be regarded as the canonical equivalence on event structures, in view of its category-theoretical characterisation as open map bisimulation with labelled partial orders as the observations~\cite{JNW96}.
HH bisimulation is defined over configuration structures and, in addition
to matching configurations and the transitions between configurations, it also
keeps a history of the matched events along matching computations.
This is achieved by means of label-preserving and order-preserving isomorphisms
between the elapsed events of the two configuration structures.
HH bisimulation and its decidability were further researched by Fr\"{o}schle in~\cite{Fro04}
and in her subsequent papers.

Various bisimulation-based equivalences weaker than HH bisimulation were studied by van Glabbeek and Goltz~\cite{vGG01} and Fecher~\cite{Fec04}.
The ones that are most relevant here are interleaving bisimulation (IB) and step bisimulation (SB).
IB is just standard bisimulation based on single event transitions.
A popular method for increasing the discriminating power of a bisimulation is to
generalise single action transitions $X\tran{a}X'$  to transitions $X\tran{\mu}X'$
where $\mu$ is a structure richer than a single action~\cite{Pom86,BC87,Che92}. 
It could be a set of events
occurring concurrently (or a multiset of action labels of these events),
a so-called ``step'', 
or even a pomset (which we shall not consider here).
SB is based on step transitions.

HH bisimulation requires the isomorphisms to be consistent under both forward and backward transitions between configurations.  It is therefore natural to look at the power of forms of IB and SB where reverse transitions are allowed as well as forward ones.  Adding reverse transitions to IB gives us what we shall call {\em reverse bisimulation} (RB).
Adding reverse transitions to SB gives us {\em reverse SB} (RSB).

RB was already investigated by Bednarczyk~\cite{Bed91}.
When there is no auto-concurrency (concurrent events with the same label), RB equivalence is finer than
many true concurrency bisimulations~\cite{vGG01} up to and including
history-preserving bisimulation~\cite{DDNM87,RT88}. The so-called {\em absorption law}~\cite{vGG01}
$$(a \Par (b + c)) + (a \Par b) + ((a + c) \Par b) =
(a \Par (b + c)) + ((a + c) \Par b)$$
is not valid for RB equivalence:
If one performs $a$ and then $b$
with the $a \Par b$ component on the left, then these must be matched
by the $a$ and then the $b$
of the $((a + c) \Par b)$ summand on the right.
(Matching it with the $a$ of $(a \Par (b + c))$ is wrong, as after this
$a$ is performed, no $c$ is possible after $a$ in $a \Par b$.)
The right hand side
can now reverse $a$ and do a $c$ (still using the same summand as all other summands
are disabled). The left hand side cannot match this.

In fact, Bednarczyk proved that, in the absence of auto-concurrency, RB equivalence has the same power as HH equivalence (on prime event structures).
We shall prove an extension of this result: RB equivalence has the same power as HH equivalence in the absence of {\em equidepth} auto-concurrency, i.e.\ when we cannot have two events with the same label occurring at the same depth within a configuration.
The {\em depth} of an event $e$ is the length of the longest causal chain of events up to and including $e$.

When auto-concurrency is present, RB is unable to distinguish such simple processes as $a \Par a$ and $a.a$.
This motivates study of RSB.\footnote{RSB was briefly mentioned by Bednarczyk; he asked the following question: is RSB as fine as HH bisimulation?
We intend to settle this question in a forthcoming extended version of the present paper.}

If we allow forward step transitions, but only single reverse transitions,
then we already can distinguish $a \Par a$ from $a.a$ very easily:
$a \Par a$ can do an $\{a,a\}$ step whereas $a.a$ cannot.
However we cannot distinguish $a \Par a$ from $(a \Par a) + a.a$.
Here the reverse steps are needed, and $a \Par a$ is not equivalent to $(a \Par a) + a.a$ for RSB;
in $(a \Par a) + a.a$ we can perform two $a$s in sequence and get to a configuration where
we cannot do a reverse step $\{a,a\}$, unlike for $a \Par a$.

We show that all the power of RSB equivalence resides in the reverse step transitions,
with the forward steps being dispensable (though of course often useful in examples).
In fact, the reverse steps can be restricted to those which are {\em homogeneous}, by which we mean that all events have the same label.
One can even restrict attention to reverse homogeneous {\em equidepth} steps, where all events have the same depth.
We also show that RSB equivalence preserves depth, in the sense that corresponding events must have the
same depth.

The paper is organised as follows. In Section~\ref{sec:cs} we define stable configuration structures and various bisimulation equivalences, including reverse forms.
Section~\ref{sec:chars} shows that reverse step bisimulation can be characterised as
reverse homogeneous step bisimulation, as
reverse homogeneous equidepth step bisimulation,
and as reverse depth-preserving bisimulation.
In Section~\ref{sec:rb hh} we show that Bednarczyk's result still holds in the absence of equidepth auto-concurrency.
We then draw some conclusions.

\section{Stable configuration structures and equivalences}\label{sec:cs}
We define configuration structures much as in~\cite{vGG01}, with the omission of the termination predicate for simplicity.
We keep as close as possible to~\cite{vGG01} in most of the definitions in this section.

We assume a set of action labels $\Act$, ranged over by $a,b,\ldots$.
\begin{defn}\label{def:cs}
A {\em configuration structure} (over an alphabet $\Act$) is a pair $\mc C = (C,\lab)$ where $C$ is a family of finite sets (configurations) and $\lab:\Union_{X \in C} X \to \Act$ is a labelling function.
\end{defn}
We denote the domain of stable configuration structures by $\Cstable$.
We use $C_{\mc C}, \lab_{\mc C}$ to refer to the two components of a configuration structure $\mc C$.
Also we let $E_{\mc C} = \Union_{X \in C} X$, the {\em events} of $\mc C$.
We let $e,\ldots$ range over events, and $E,F,\ldots$ over sets of events.
\begin{defn}
A configuration structure $\mc C = (C,\lab)$ is {\em stable} if it is
\begin{itemize}
\item
rooted: $\emptyset \in C$;
\item
connected: $\emptyset \neq X \in C$ implies $\exists e \in X: X \setminus\{e\} \in C$;
\item
closed under bounded unions: if $X,Y,Z \in C$ then $X \union Y \subseteq Z$ implies $X \union Y \in C$;
\item
closed under bounded intersections: if $X,Y,Z \in C$ then $X \union Y \subseteq Z$ implies $X \inter Y \in C$.
\end{itemize}
\end{defn}
Any stable configuration structure is the set of configurations of a 
stable event structure~\cite[Theorem 5.3]{vGG01}.
\begin{figure}
\psfrag{C}{$\mc C$}
\psfrag{D}{$\mc D$}
\psfrag{0}{$\emptyset$}
\psfrag{{0}}{$\{0\}$}
\psfrag{{1}}{$\{1\}$}
\psfrag{{0,1}}{$\{0,1\}$}
\psfrag{{0,b}}{$\{0,b\}$}
\psfrag{{1,b}}{$\{1,b\}$}
\psfrag{{0,1,b}}{$\{0,1,b\}$}
\centering
\includegraphics[trim = 0in 0.7in 0in 0in, clip,width=4in]{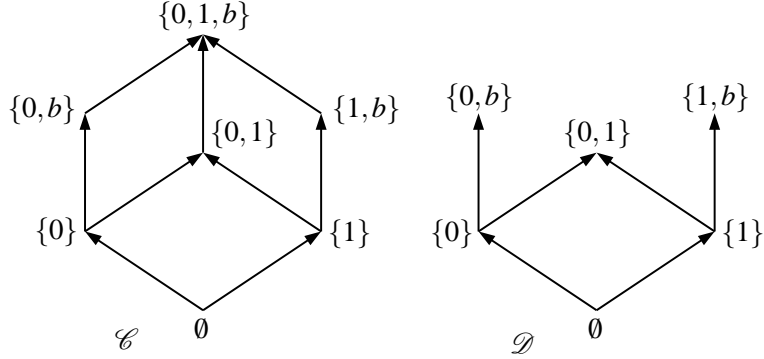}
\caption{Non-stable and stable configuration structures.}\label{config2}
\end{figure}
In Figure~\ref{config2} we give two example configuration structures derived from examples of Winskel~\cite{Win87}.
In each of $\mc C, \mc D$ the labelling can be taken to be the identity function.
Configuration structure $\mc C$ models a ``parallel switch'' where events $0$ or $1$ can light the bulb $b$.
If both $0$ and $1$ occur then $b$ can also happen.
This is {\em inclusive} ``or'' causation.
We see that $\mc C$ is not stable, since it is not closed under bounded intersections:
$\{0,b\} \union \{1,b\}$ is bounded by $\{0,1,b\}$, but $\{0,b\} \inter \{1,b\} = \{b\} \notin C_{\mc C}$.
By contrast, $\mc D$ is  stable; it models a switch where the bulb can be lit by either $0$ or $1$, but not both,
i.e.\ {\em exclusive} ``or'' causation.

Configuration structures have associated notions of causal orderings on events and concurrency between events:
\begin{defn}
Let $\mc C = (C,\lab) \in \Cstable$, and let $X \in C$.
\begin{itemize}
\item
Causality: $d \leq_X e$ iff for all $Y \in C$ with $Y \subseteq X$ we have $e \in Y$ implies $d \in Y$.
Furthermore $d <_X e$ iff $d \leq_X e$ and $d \neq e$.
\item
Concurrency: $d \co_X e $ iff $d \not <_X e$ and $e \not <_X d$.
\end{itemize}
\end{defn}

It is shown in~\cite{vGG01} that $<_X$ is a partial order and that the sub-configurations of $X$ are precisely those subsets $Y$ which are left-closed w.r.t.\ $<_X$, i.e.\ if $d <_X e \in Y$ then $d \in Y$.
Furthermore, if $X,Y \in C$ with $Y \subseteq X$, then ${<_Y} = {<_X \res Y}$.

\begin{rem}
{\em Prime} event structures form a proper subclass of stable event structures.
We do not give the definition here.
However we remark that the configuration structures associated with prime event structures are got by strengthening 
the ``closed under bounded intersections'' condition of Definition~\ref{def:cs} to closure under intersections:
if $X,Y \in C$ then $X \inter Y \in C$~\cite{vG96}.
In Figure~\ref{config2}, $\mc D$ is not prime, since it is not closed under intersections:
$\{0,b\} \inter \{1,b\} = \{b\} \notin C_{\mc D}$.
Thus prime event structures do not allow ``or'' causation;
to model the switch as a prime event structure we would have to model the lighting of the bulb as two separate events, one caused by $0$ and the other by $1$.
\end{rem}

We now define various notions of equivalence between configuration structures.
We start by defining the most basic labelled transition relation, on single events.  We also use a reverse transition relation, with a wavy arrow, which simply inverts the standard forward version.

\begin{defn}
Let $\mc C = (C,\lab) \in \Cstable$ and let $a \in \Act$.
We let $X \tranc a X'$ iff $X,X' \in C$, $X \subseteq X'$ and $X' \setminus X = \{e\}$ with $\lab(e) = a$.
Also $X \rtranc a X'$ iff $X' \tranc a X$.
\end{defn}

\begin{defn}[\cite{vGG01}]
Let $\mc C, \mc D \in \Cstable$.
A relation $R \subseteq C_{\mc C} \cross C_{\mc D}$ is an {\em interleaving bisimulation} (IB) between $\mc C$ and $\mc D$ if $(\emptyset,\emptyset) \in R$ and if $(X,Y) \in R$ then for $a \in \Act$
\begin{itemize}
\item
if $X \tranc a X'$ then $\exists Y'.\ Y \trand a Y'$ and $(X',Y') \in R$;
\item
if $Y \trand a Y'$ then $\exists X'.\ X \tranc a X'$ and $(X',Y') \in R$.
\end{itemize}
We say that $\mc C$ and $\mc D$ are IB equivalent ($\mc C \ibeqt \mc D$) iff there is an IB between $\mc C$ and $\mc D$.
\end{defn}
For a set of events $E$, let $\lab(E)$ be the multiset of labels of events in $E$.
We define a {\em step} transition relation where concurrent events are executed in a single step:
\begin{defn}
Let $\mc C = (C,\lab) \in \Cstable$ and let $A \in \Nat^\Act$ ($A$ is a multiset over $\Act$).
We let $X \tranc A X'$ iff $X,X' \in C$, $X \subseteq X'$, and $X' \setminus X = E$ with $d \co_{X'} e$ for all $d,e \in E$ and $\lab(E) = A$.
\end{defn}

\begin{defn}[\cite{Pom86,vGG01}]
Let $\mc C, \mc D \in \Cstable$.
A relation $R \subseteq C_{\mc C} \cross C_{\mc D}$ is a {\em step bisimulation} (SB) between $\mc C$ and $\mc D$ if $(\emptyset,\emptyset) \in R$ and if $(X,Y) \in R$ then for $A \in \Nat^\Act$
\begin{itemize}
\item
if $X \tranc A X'$ then $\exists Y'.\ Y \trand A Y'$ and $(X',Y') \in R$;
\item
if $Y \trand A Y'$ then $\exists X'.\ X \tranc A X'$ and $(X',Y') \in R$.
\end{itemize}
We say that $\mc C$ and $\mc D$ are SB equivalent ($\mc C \sbeqt \mc D$) iff there is an SB between $\mc C$ and $\mc D$.
\end{defn}

Hereditary history-preserving bisimulation was defined in~\cite{Bed91}, where it is called hereditary {\em strong} history-preserving (HH) bisimulation.
\begin{defn}\label{def:hh}
Let $\mc C, \mc D \in \Cstable$.
A relation $R \subseteq C_{\mc C} \cross C_{\mc D} \cross \pow(E_{\mc C} \cross E_{\mc D})$ is a {\em hereditary history-preserving (HH) bisimulation} between $\mc C$ and $\mc D$ if $(\emptyset,\emptyset,\emptyset) \in R$ and if $(X,Y,f) \in R$ and $a \in \Act$
\begin{itemize}
\item
$f$ is an isomorphism between $(X, <_X,\lab_{\mc C}\res X)$ and $(Y, <_Y,\lab_{\mc D}\res X)$;
\item
if $X \tranc a X'$ then $\exists Y',f'.\ Y \trand a Y'$, $(X',Y',f') \in R$ and $f'\res X = f$;
\item
if $Y \trand a Y'$ then $\exists X',f'.\ X \tranc a X'$, $(X',Y',f') \in R$ and $f'\res X = f$;
\item
if $X \rtranc a X'$ then $\exists Y',f'.\ Y \rtrand a Y'$, $(X',Y',f') \in R$ and $f\res X' = f'$.
\end{itemize}
We say that $\mc C$ and $\mc D$ are HH equivalent ($\mc C \hheqt \mc D$) iff there is an HH bisimulation between $\mc C$ and $\mc D$.
\end{defn}
Note that we do not need a clause for $Y \rtranc a Y'$ in Definition~\ref{def:hh}, since it is entailed by the clause for $X \rtranc a X'$ (for a given $Y'$, we have that $X',f'$ are fully determined, given $f\res X' = f'$ and $f'(X') = Y'$).

\begin{prop}[\cite{vGG01}]\label{prop:hh-sb-ib}
On stable configuration structures,
${\hheqt} \subsetneq {\sbeqt} \subsetneq {\ibeqt}$.
\qed
\end{prop}
We give two examples to show that the inclusions in Proposition~\ref{prop:hh-sb-ib} are proper.
Here and subsequently we use a CCS-like notation to refer to simple configuration structures.
\begin{example}\label{ex:hh-sb-ib}
\begin{enumerate}
\item
IB equivalence is insensitive to auto-concurrency:
$a \Par a = a.a$ holds for $\ibeqt$, but not for $\sbeqt$.
\item
$a \Par a = (a \Par a) + a.a$ holds for $\sbeqt$, but not for $\hheqt$.
\end{enumerate}
\end{example}
We now define an enhancement of IB with reverse transitions.  This was defined in~\cite{Bed91}, where it is called {\em back \& forth} bisimulation ($\sim_{b\& f}$).
It was called forward-reverse (FR) bisimulation in~\cite{PU07}.
\begin{defn}
Let $\mc C, \mc D \in \Cstable$.
A relation $R \subseteq C_{\mc C} \cross C_{\mc D}$ is a {\em reverse} bisimulation (RB) between $\mc C$ and $\mc D$ if
it is an IB and if $(X,Y) \in R$ then for $a \in \Act$
\begin{itemize}
\item
if $X \rtranc a X'$ then $\exists Y'.\ Y \rtrand a Y'$ and $(X',Y') \in R$;
\item
if $Y \rtrand a Y'$ then $\exists X'.\ X \rtranc a X'$ and $(X',Y') \in R$.
\end{itemize}
We say that $\mc C$ and $\mc D$ are RB equivalent ($\mc C \rbeqt \mc D$) iff there is an RB between $\mc C$ and $\mc D$.
\end{defn}
De Nicola, Montanari
and Vaandrager also investigated ``back \& forth'' bisimulations~\cite{DNMV90,DNV90},
but their relations were defined over computations (paths) rather than states
(for example, in the process $a \Par b$, after performing $a$ followed by $b$, one can only reverse immediately on $b$, and not $a$).
As a result,
in the absence of $\tau$ actions, 
the distinguishing power of these bisimulations is that of IB~\cite{DNMV90};
hence, it is lower than that of RB.

We next define SB with added reverse steps.  This was briefly mentioned in~\cite{Bed91}, where it is called {\em multi-step} back \& forth bisimulation ($\sim_{\mu b\& f}$).
\begin{defn}
Let $\mc C, \mc D \in \Cstable$.
A relation $R \subseteq C_{\mc C} \cross C_{\mc D}$ is a {\em reverse} SB (RSB) between $\mc C$ and $\mc D$ if
it is an SB and if $(X,Y) \in R$ then for $A \in \Nat^\Act$
\begin{itemize}
\item
if $X \rtranc A X'$ then $\exists Y'.\ Y \rtrand A Y'$ and $(X',Y') \in R$;
\item
if $Y \rtrand A Y'$ then $\exists X'.\ X \rtranc A X'$ and $(X',Y') \in R$.
\end{itemize}
We say that $\mc C$ and $\mc D$ are RSB equivalent ($\mc C \rsbeqt \mc D$) iff there is an RSB between $\mc C$ and $\mc D$.
\end{defn}

We give some further examples to show the differences between the various equivalences:
\begin{example}\label{ex:simple}
\begin{enumerate}
\item\label{ab+ba}
If $a \neq b$ then an interleaving law $a \Par b = a.b + b.a$ holds for $\ibeqt$, but not for $\rbeqt$ or $\sbeqt$.
\item\label{aa}
RB equivalence is insensitive to auto-concurrency:
$a \Par a = a.a$ holds for $\rbeqt$, but not for $\sbeqt$.
\item\label{abs}
The Absorption Law~\cite{vGG01}
$(a \Par (b + c)) + (a \Par b) + ((a + c) \Par b) =
(a \Par (b + c)) + ((a + c) \Par b)$
holds for $\sbeqt$, but not for $\rbeqt$.
\end{enumerate}
\end{example}

\begin{prop}\label{prop:inclusions}
On stable configuration structures,
\Comment{
\begin{enumerate}
\item\label{hh-rsb}
${\hheqt} \subseteq {\rsbeqt}$
\item\label{rsb-rb}
${\rsbeqt} \subsetneq {\rbeqt}$
\item\label{rsb-sb}
${\rsbeqt} \subsetneq {\sbeqt}$
\item\label{rb-ib}
${\rbeqt} \subsetneq {\ibeqt}$
\item\label{rb-sb}
${\rbeqt} \not\subseteq {\sbeqt}$
\item\label{sb-rb}
${\sbeqt} \not\subseteq {\rbeqt}$
\end{enumerate}
} 
\[\begin{array}{llcll}
1. & {\hheqt} \subseteq {\rsbeqt}
&\qquad&
4. & {\rbeqt} \subsetneq {\ibeqt}
\\
2. & {\rsbeqt} \subsetneq {\rbeqt}
&\qquad&
5. & {\rbeqt} \not\subseteq {\sbeqt}
\\
3. & {\rsbeqt} \subsetneq {\sbeqt}
&\qquad&
6. & {\sbeqt} \not\subseteq {\rbeqt}
\\
\end{array}
\]
\end{prop}
\begin{proof}
(1) is shown in~\cite{Bed91}.
It is also essentially shown in~\cite{vGG01} (see Proposition~9.1 and the remarks after Definition 9.6).

The remaining parts follow from the definitions and Example~\ref{ex:simple}.
\Comment{
\begin{description}
\item
(\ref{hh-rsb})
This is shown in~\cite{Bed91}.
\item
(\ref{rsb-rb})
It is clear from the definitions that ${\rsbeqt} \subseteq {\rbeqt}$.
That ${\rbeqt} \not\subseteq {\rsbeqt}$ follows from 
Example~\ref{ex:simple}(\ref{aa}).
\item
(\ref{rsb-sb})
It is clear from the definitions that ${\rsbeqt} \subseteq {\sbeqt}$.
That ${\sbeqt} \not\subseteq {\rsbeqt}$ follows from 
Example~\ref{ex:simple}(\ref{abs}).
\item
(\ref{rb-ib})
It is clear from the definitions that ${\rbeqt} \subseteq {\ibeqt}$.
That ${\ibeqt} \not\subseteq {\rbeqt}$ follows from 
Example~\ref{ex:simple}(\ref{ab+ba}) (or (\ref{abs})).

\item
(\ref{rb-sb})
By Example~\ref{ex:simple}(\ref{aa}).

\item
(\ref{sb-rb})
By Example~\ref{ex:simple}(\ref{abs}).

\end{description}
} 
\end{proof}
As stated in the Introduction, it remains an open question whether ${\rsbeqt} = {\hheqt}$.
Figure~\ref{hierarchy2} illustrates Proposition~\ref{prop:inclusions}.
Inclusions are represented by arrows.
\begin{figure}
\psfrag{ib}{ib}
\psfrag{rib}{rb}
\psfrag{sb}{sb}
\psfrag{rsb}{rsb}
\psfrag{hh}{hh}
\centering
\includegraphics[width=0.9in]{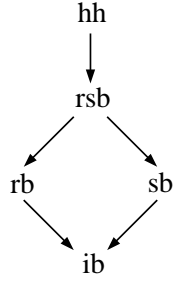}
\caption{Equivalences discussed in Section~\ref{sec:cs}.}\label{hierarchy2}
\end{figure}

\section{Characterisations of Reverse Step Bisimulation}\label{sec:chars}
We show that RSB equivalence can be characterised in three ways:
\begin{enumerate}
\item
as reverse homogeneous step bisimulation (RHSB) equivalence, where forward steps are not used, and reverse steps are homogeneous, i.e.\ all events have the same label
(Section~\ref{subsec:rhsb});
\item
as reverse depth-respecting bisimulation (RDB) equivalence, where events are matched on depth as well as on label
(Section~\ref{subsec:rdb}).
\item
as reverse homogeneous equidepth step bisimulation (RHESB) equivalence, which is the same as RHSB equivalence, with the additional proviso that reverse steps are equidepth, i.e.\ all events have the same depth
(Section~\ref{subsec:rdb} also).
\end{enumerate}
\subsection{Reverse Homogeneous Step Bisimulation}\label{subsec:rhsb}

As we already observed in the Introduction, if we allow forward step transitions, but only single reverse transitions, then we have a strictly weaker notion than RSB equivalence, which, although stronger than RB equivalence (since it distinguishes $a \Par a$ from $a.a$), is unable to distinguish $a \Par a$ from $(a \Par a) + a.a$.


Let us say that a set of events is {\em homogeneous} if all events have the same label.  Similarly, a multiset of labels is homogeneous if all labels are the same.  We next show that the power of SR bisimulation lies in the reverse steps, and in particular reverse homogeneous steps, so that forward steps are in fact superfluous.  
\begin{defn}\label{def:RHSB}
Let $\mc C,\mc D \in \Cstable$.
A relation $R \subseteq C_{\mc C}\times C_{\mc D}$
is a {\em reverse homogeneous step bisimulation (RHSB)}
between $\mc C$ and $\mc D$ if $R(\emptyset, \emptyset)$ and
whenever $R(X,Y)$ then
\begin{itemize}
\item
if $X \tranc a X'$ then $\exists Y'.\ Y \trand a Y'$ 
and $R(X',Y')$;
\item
if $Y \trand a Y'$ then $\exists X'.\ X \tranc a X'$ 
and $R(X',Y')$;
\item
if $X \rtranc A X'$, where $A$ is homogeneous, then $\exists Y'.\ Y \rtrand A Y'$ 
and $R(X',Y')$;
\item
if $Y \rtrand A Y'$, where $A$ is homogeneous, then $\exists X'.\ X \rtranc A X'$ 
and $R(X',Y')$.
\end{itemize}
We say that $\mc C$ and $\mc D$ are RHSB equivalent ($\mc C \rhsbeqt \mc D$) iff there is an RHSB between $\mc C$ and $\mc D$.
\end{defn}
Clearly, any RSB is an RHSB.  We shall show the converse, so that ${\rhsbeqt} = {\rsbeqt}$ (Theorem~\ref{thm:RHSB=RSB}).

First we need some lemmas.
\begin{lemma}\label{lem:labels}
Let $\mc C$ and $\mc D \in \Cstable$ and let $R$ be an RB between $\mc C$ and $\mc D$.
If $R(X,Y)$ then $\lab(X) = \lab(Y)$.
\end{lemma}
\begin{proof}
Suppose that $R(X,Y)$ and suppose that $\lab(X) = A$.
By connectedness, there are transitions $X \rtranc {a_1} \cdots \rtranc {a_n} \emptyset$ with $A = \{a_1,\ldots,a_n\}$.
Therefore $Y \rtrand {a_1} \cdots \rtrand {a_n} Y'$, for some $Y'$.
Hence $\lab(X) \subseteq \lab(Y)$.
Symmetrically, $\lab(Y) \subseteq \lab(X)$.
\end{proof}
Note that Lemma~\ref{lem:labels} would not hold for IBs: for instance, there is an IB between $a+b$ and itself which, while necessarily including the identity relation on configurations, also relates the configuration resulting after performing (the event labelled) $a$ with that after performing $b$.

For any configuration $X$, 
let $\mn X$ denote its set of minimal elements (w.r.t.\ $<_X$);
$\mn X$ is, of course, also a configuration.
Note that $X,Y$ are configurations and $X \subseteq Y$ then $\mn X \subseteq \mn Y$.

\begin{lemma}\label{lem:rhsb min}
Let $\mc C,\mc D \in \Cstable$ be related by RHSB $R$.
If $R(X,Y)$ then we have $R(\mn X,\mn {Y})$. 
\end{lemma}
\begin{proof}
Suppose $R(X,Y)$.  Then there are $a_1,\ldots,a_n$ such that
$X \rtran{a_1} \cdots \rtran{a_n} \mn X$.
Let $Y'$ be such that $Y \rtran{a_1} \cdots \rtran{a_n} Y'$ and $R(\mn X,Y')$.
Then $\lab(\mn X) = \lab(Y')$ by Lemma~\ref{lem:labels}.  We show that $Y' = \mn {Y}$.

Now we use the reverse homogeneous steps.
Let $A = \lab(\mn X)$.  Take any $a \in A$, and let $A_a$ be the multiset of $a$s in $A$.
Then $\mn X \rtran {b_1} \cdots \rtran {b_n} \rtran {A_a} \emptyset$ for some $b_1,\ldots,b_n \neq a$.
Hence $Y' \rtran {b_1} \cdots \rtran {b_n} \rtran {A_a} \emptyset$.
This tells us that all events labelled with $a$ in $Y'$ are minimal.
Hence all events in $Y'$ are minimal, since $a$ was arbitrary.  So $Y' \subseteq \mn{Y}$.
Hence $\lab(\mn X) \subseteq \lab(\mn {Y})$.
Symmetrically we can establish $\lab(\mn {Y}) \subseteq \lab(\mn X)$.
So $\lab(\mn {Y}) = \lab(\mn X)$.
Hence $Y' = \mn{Y}$ and $R(\mn X,\mn {Y})$ as required.
\end{proof}

We now define the ``lifting'' of a configuration structure with respect to a configuration $M$:

\begin{defn}\label{def:cs lift}
Let $\mc C = (C,\lab) \in \Cstable$ and let $M \in C$.
Define $\mc C_M = (C_M,\lab_M)$ where
$C_M = \{X \setminus M : M \subseteq X \in C, \mn{X} = \mn M \}$ and
$\lab_M = \lab \res \Union_{Y \in C_M} Y$.
\end{defn}

\begin{lemma}\label{lem:cs lift}
Let $\mc C = (C,\lab) \in \Cstable$ and let $M \in C$.
Then $\mc C_M \in \Cstable$.
\qed
\end{lemma}

\begin{lemma}\label{lem:rhsb lift}
Let $\mc C,\mc D \in \Cstable$ be related by RHSB $R$.
Let $M \in C_{\mc C}$ be such that $\mn M = M$.
Similarly, let $N \in C_{\mc D}$ be such that $\mn N = N$.
Suppose also that $R(M,N)$.
Define $R_{M,N}$ by
\[
R_{M,N} =
\{(X \setminus M,Y \setminus N) :
R(X,Y),\ \mn X = M,\ \mn Y = N\}
\]
Then $R_{M,N}$ is an RHSB between $\mc C_M$ and $\mc D_N$.
\end{lemma}
\begin{proof}
Write $R_{M,N}$ as $R'$ for short.
Certainly $R'(\emptyset,\emptyset)$, since $R(M,N)$.
Suppose $R'(X \setminus M, Y \setminus N)$ with $R(X,Y)$, $\mn {X} = M$, $\mn {Y} = N$.

Forwards: Suppose $X \setminus M \tran a_{\mc C_M} X' \setminus M$,
where $X' \in C_{\mc C}$ and $\mn {X'} = M$.
Then $X \tranc a X'$.
Hence there is $Y'$ such that $Y \trand a Y'$ and $R(X',Y')$.
We need to know that $\mn{Y'} = N$.
Using Lemma~\ref{lem:rhsb min} we have $R(M,\mn {Y'})$.
By Lemma~\ref{lem:labels}, since $R(M,N)$ we have $\lab(M) = \lab(N)$.
Also since $R(M,\mn {Y'})$ we have $\lab(M) = \lab(\mn {Y'})$.
So $\lab(N) = \lab(\mn {Y'})$.
Since $N =\mn Y \subseteq \mn {Y'}$ we deduce
$\mn {Y'} = N$.
Now $Y \setminus N \tran a_{\mc D_N} Y' \setminus N$ and $R'(X' \setminus M,Y' \setminus N)$ as required.

Reverse: Suppose $X \setminus M \rtran A_{\mc C_M} X' \setminus M$,
where $X' \in C_{\mc C}$ and $\mn {X'} = M$.
Then $X \rtranc A X'$.
Hence there is $Y'$ such that $Y \rtrand A Y'$ and $R(X',Y')$.
We need to know that $\mn{Y'} = N$.
By Lemma~\ref{lem:rhsb min} we have $R(M,\mn{Y'})$.
So by Lemma~\ref{lem:labels}, $\lab(\mn{Y'}) = \lab(M) = \lab(N)$.
Since $Y' \subseteq Y$, $\mn{Y'} \subseteq \mn{Y} = N$.
Hence $\mn{Y'} = N$.
Now $Y \setminus N \rtran A_{\mc D_{N}} Y' \setminus N$ and $R'(X' \setminus M,Y' \setminus N)$ as required.
\end{proof}

\begin{lemma}\label{lem:fwd step}
Let $\mc C,\mc D \in \Cstable$ be related by RHSB $R$.
Suppose $R(X,Y)$ and $X \tranc A X'$.
Then there is $Y'$ such that $Y \trand A Y'$ and $R(X',Y')$.
\end{lemma}
\begin{proof}
The proof is partially inspired by Fecher's proof that a weak history-preserving bisimulation is a step bisimulation~\cite{Fec04}.

We proceed by induction on $\card {X'}$.

Base case: $\card {X'} = 0$.  Then $Y' = \emptyset$ will do, trivially.

Induction step.  Notice that $\mn X \subseteq \mn {X'}$.
There are two cases:

(1)
There is $e \in \mn {X'} \setminus \mn X$.  Let $\lab(e) = a$.
Then $X \tran {A \setminus \{a\}} X' \setminus \{e\}$.
By induction there is $Y''$ such that $Y \tran {A \setminus \{a\}} Y''$ and $R(X' \setminus \{e\},Y'')$.
Now $X' \setminus \{e\} \tran a X'$.
So there is $Y'$ such that $Y'' \tran a Y'$ and $R(X',Y')$.
Let $e'$ be the single element in $Y' \setminus Y''$.
Now $\mn{X' \setminus \{e\}} = \mn {X'} \setminus \{e\}$ (and $e \in \mn {X'}$).
So $\card{\mn{X' \setminus \{e\}}} < \card{ \mn {X'}}$.
By Lemmas~\ref{lem:rhsb min} and~\ref{lem:labels}, $\lab(\mn{X' \setminus \{e\}}) = \lab(\mn{Y' \setminus \{e'\}})$
and $\lab(\mn {X'}) = \lab(\mn{Y'})$.
So $\card{\mn{Y' \setminus \{e'\}}} < \card{ \mn {Y'}}$
and $\mn{Y' \setminus \{e'\}} \subsetneq \mn {Y'}$.
It follows that 
$e' \in \mn{Y'}$.  Hence $e'$ is concurrent with all events in $Y' \setminus Y$, and
$Y \tran A Y'$ as required.

(2)
$\mn {X'} = \mn X$.
Let $M = \mn X$, $N = \mn Y$.
By Lemma~\ref{lem:rhsb min} we have $R(M,N)$.

Let configuration structures $\mc C_M$ and $\mc D_N$ be as in Definition~\ref{def:cs lift}.
Let $R'$ be the RHSB between $\mc C_M$ and $\mc D_N$ of Lemma~\ref{lem:rhsb lift}.

We have $R'(X \setminus M, Y \setminus N)$ and $X \setminus M \tran A_{\mc C_M} X' \setminus M$.
Clearly $M \neq \emptyset$, since $\card {X'} > 0$.
So by induction there is $Y'$ such that $Y \setminus N \tran A_{\mc D_N} Y' \setminus N$
with $R'(X' \setminus M, Y' \setminus N)$.  So $Y \tran A_{\mc D} Y'$ and $R(X',Y')$ as required.
\end{proof}

We use the same method as Lemma~\ref{lem:fwd step} to show:
\begin{lemma}\label{lem:rev step}
Let $\mc C,\mc D \in \Cstable$ be related by RHSB $R$.
Suppose $R(X,Y)$ and $X \rtranc A X'$.
Then there is $Y'$ such that $Y \rtrand A Y'$ and $R(X',Y')$.
\end{lemma}
\begin{proof}
Much as the proof of Lemma~\ref{lem:fwd step}.
We proceed by induction on $\card X$, going into much the same two cases.

In (1) note that we start by reversing from $X$ in a single event transition using an element $e \in \mn X \setminus \mn {X'}$, and we then do an $A \setminus \{a\}$ reverse step.
If we did the $A \setminus \{a\}$ reverse step followed by the single $a$ reverse transition, then on the $\mc D$ side we leave open the possibility that $e'$ causes the remaining events of $Y' \union \{e'\} \tran {A \setminus \{a\}} Y$.

In (2) we define $R'$ in exactly the same way, and everything works much as before.
\end{proof}
Combining:

\begin{theorem}\label{thm:RHSB=RSB}
On stable configuration structures,
${\rhsbeqt} = {\rsbeqt}$\ .
\end{theorem}
\begin{proof}
Clearly any RSB is an RHSB.
Any RHSB is an RSB, by Lemmas~\ref{lem:fwd step} and~\ref{lem:rev step}.
\end{proof}

\subsection{Depth-respecting Bisimulations}\label{subsec:rdb}
We introduce various new notions of equivalence, which take into account the {\em depth} of events within a configuration.
In particular we show that RSB equivalence can be characterised
as reverse depth-respecting bisimulation (RDB) equivalence, where events are matched on depth as well as on label,
and as reverse homogeneous equidepth step bisimulation (RHESB) equivalence,
which is a variant of RHSB equivalence in which all events in a reverse step have the same depth.

We start by defining depth:
\begin{defn}\label{def:depth}
Let $\mc C = (C,\lab) \in \Cstable$, and let $X \in C$, $e \in X$.
The {\em depth} of $e$ w.r.t.\ $X$ (and implicitly $\mc C$) is given by
\[
\depth X e \Defeq \left\{ \begin{array}{ll}
1 & \mbox{if } e \mbox{ is minimal in } X
\\
\max\{\depth X {e'} : e' <_X e\} +1 & \mbox{otherwise}
\end{array} \right.
\]
\end{defn}
The depth of an event $e$ is the length of the longest causal chain in $X$ up to and including $e$.
Clearly, if $e <_X e'$ then $\depth X e <_X \depth X {e'}$.
Note that if $X,Y \in C$ and $e \in X \inter Y$, then it is not necessarily the case that $\depth X e = \depth Y e$
(due to the fact that a single event can have different possible sets of causes).
However if $X \union Y \subseteq Z$ for some $Z \in C$ then $\depth X e = \depth Y e$.

\begin{defn}
Let $\mc C = (C,\lab) \in \Cstable$ and let $a \in \Act$, $ k \in \Nat$.
We let $X \dtranc a k X'$ iff $X,X' \in C$, $X \subseteq X'$ and $X' \setminus X = \{e\}$
with $\lab(e) = a$, $\depth {X'} e = k$.
Also $X \rdtranc a k X'$ iff $X' \dtranc a k X$.
\end{defn}

\begin{defn}
Let $\mc C, \mc D \in \Cstable$.
A relation $R \subseteq C_{\mc C} \cross C_{\mc D}$ is a {\em depth-respecting} bisimulation (DB) between $\mc C$ and $\mc D$ if $(\emptyset,\emptyset) \in R$ and
if $(X,Y) \in R$ then for $a \in \Act$ and $k \in \Nat$
\begin{itemize}
\item
if $X \dtranc a k X'$ then $\exists Y'.\ Y \dtrand a k Y'$ and $(X',Y') \in R$;
\item
if $Y \dtrand a k Y'$ then $\exists X'.\ X \dtranc a k X'$ and $(X',Y') \in R$.
\end{itemize}
We say that $\mc C$ and $\mc D$ are DB equivalent ($\mc C \dbeqt \mc D$) iff there is a DB between $\mc C$ and $\mc D$.
\end{defn}
\begin{example}\label{ex:sb-db}
On stable configuration structures,
\begin{enumerate}
\item\label{sb-db}
$a \Par b = (a \Par b) + a.b$ holds for $\sbeqt$, but not for $\dbeqt$;
\item\label{abs-db}
The Absorption Law (Example~\ref{ex:simple}(\ref{abs})) holds for $\dbeqt$. 
\end{enumerate}
\end{example}
\begin{prop}\label{prop:db-sb}
On stable configuration structures,
${\dbeqt} \subsetneq {\sbeqt}$.
\end{prop}
\begin{proof}
Suppose that $\mc C \dbeqt \mc D$ via DB $R$.
We show that $R$ is an SB.
Let $A \in \Nat^{\Act}$, and suppose $R(X,Y)$.

Assume $X \tranc A X'$.  Let $E = X' \setminus X$ and let $\{e_1,\ldots,e_n\}$ be an enumeration of $E$ in non-increasing order of depth w.r.t.\ $X'$,
i.e., letting $\depth {X'}{e_i} = k_i$ we have $k_i \geq k_j$ for $i < j \leq n$.
Let $\lab_{\mc C}(e_i) = a_i$ ($i \leq n$).
Then $X = X_0 \dtranc {a_1}{k_1} X_1 \cdots \dtranc {a_n}{k_n} X_n = X'$.
So $Y = Y_0 \dtrand {a_1}{k_1} Y_1 \cdots \dtrand {a_n}{k_n} Y_n = Y'$
for some $Y_1,\ldots,Y_n$ such that $R(X_i,Y_i)$ ($i \leq n$).
Let $e'_i = Y_i \setminus Y_{i-1}$ ($i = 1,\ldots, n$).
The $e'_i$ must all be pairwise concurrent:
if $i < j$ then $e'_j <_{Y'} e'_i$ is impossible
since $e'_j \notin Y_i$ and $Y_i$ is left-closed;
also, $e'_i <_{Y'} e'_j$ is impossible since $\depth {Y'}{e_i} \geq \depth {Y'}{e_j}$.
Hence $Y \trand A Y'$ with $R(X',Y')$, as required.

By symmetry, we also have that if $Y \tranc A Y'$ then $X \tranc A X'$ for some $X'$ such that $R(X',Y')$.
This shows that ${\dbeqt} \subseteq {\sbeqt}$.
The inclusion is proper by Example~\ref{ex:sb-db}(\ref{sb-db}).
\end{proof}

\begin{defn}
Let $\mc C, \mc D \in \Cstable$.
A relation $R \subseteq C_{\mc C} \cross C_{\mc D}$ is a {\em reverse depth-respecting} bisimulation (RDB) between $\mc C$ and $\mc D$ if it is a DB and
if $(X,Y) \in R$ then for $a \in \Act$ and $k \in \Nat$
\begin{itemize}
\item
if $X \rdtranc a k X'$ then $\exists Y'.\ Y \rdtrand a k Y'$ and $(X',Y') \in R$;
\item
if $Y \rdtrand a k Y'$ then $\exists X'.\ X \rdtranc a k X'$ and $(X',Y') \in R$.
\end{itemize}
We say that $\mc C$ and $\mc D$ are RDB equivalent ($\mc C \rdbeqt \mc D$) iff there is an RDB between $\mc C$ and $\mc D$.
\end{defn}

\begin{prop}\label{prop:rdb-rsb}
On stable configuration structures,
${\rdbeqt} \subseteq {\rsbeqt}$.
\end{prop}
\begin{proof}
We show that any RDB is an RSB, by much the same method as the proof of Proposition~\ref{prop:db-sb}.
\end{proof}
We shall later (Theorem~\ref{thm:RSB=RDB}) show that the converse of Proposition~\ref{prop:rdb-rsb} also holds.

We now define {\em equidepth} step transitions, i.e.\ step transitions where all events have the same depth:
\begin{defn}
Let $\mc C = (C,\lab) \in \Cstable$ and let $A \in \Nat^\Act$.
We let $X \dtranc A = X'$ iff $X \tranc A X'$ and all events in $X' \setminus X$ have the same depth.
Also $X \rdtranc A = X'$ iff $X' \dtranc A = X$.
\end{defn}

\begin{defn}\label{def:RHESB}
Let $\mc C, \mc D \in \Cstable$.
A relation $R \subseteq C_{\mc C} \cross C_{\mc D}$ is a {\em reverse homogeneous equidepth step} bisimulation (RHESB) between $\mc C$ and $\mc D$ if $R(\emptyset, \emptyset)$ and
whenever $R(X,Y)$ then
\begin{itemize}
\item
if $X \tranc a X'$ then $\exists Y'.\ Y \trand a Y'$ 
and $R(X',Y')$;
\item
if $Y \trand a Y'$ then $\exists X'.\ X \tranc a X'$ 
and $R(X',Y')$;
\item
if $X \rdtranc A = X'$, where $A$ is homogeneous, then $\exists Y'.\ Y \rdtrand A = Y'$ and $(X',Y') \in R$;
\item
if $Y \rdtrand A = Y'$, where $A$ is homogeneous, then $\exists X'.\ X \rdtranc A = X'$ and $(X',Y') \in R$.
\end{itemize}
We say that $\mc C$ and $\mc D$ are RHESB equivalent ($\mc C \rhesbeqt \mc D$) iff there is an RHESB between $\mc C$ and $\mc D$.
\end{defn}
Notice that it is not obvious that an RSB is an RHESB, or (conversely) that an RHESB is an RHSB, because of the equidepth condition.

\begin{prop}\label{prop:rdb-rhesb}
On stable configuration structures,
${\rdbeqt} \subseteq {\rhesbeqt}$.
\end{prop}
\begin{proof}
Let $R$ be an RDB between $\mc C, \mc D \in \Cstable$.
We show that $R$ is an RHESB.
Suppose $R(X,Y)$.  If $X \tranc a X'$ then $X \dtranc a k X'$ for some $k$.
So there is $Y'$ such that $R(X',Y')$ and $Y \dtrand a k Y'$.
Then $Y \trand a Y'$ as required.
The case for $Y \trand a Y'$ is similar.

If $X \rdtranc A = X'$, where $A = \{a,\ldots,a\}$ is homogeneous, then there is $k$ such that
$X \rdtranc a k \cdots \rdtranc a k X'$.
So there is $Y'$ such that $Y \rdtrand a k \cdots \rdtrand a k Y'$ with $R(X',Y')$.
But then clearly $Y \rdtrand A = Y'$ as required.
The case for $Y \rdtrand A = Y'$ is similar.
\end{proof}
We shall later (Theorem~\ref{thm:RSB=RDB}) show that the converse of Proposition~\ref{prop:rdb-rhesb} also holds.

First we need some further results.

\begin{lemma}\label{lem:rhesb min}
Let $\mc C,\mc D \in \Cstable$.
Suppose that $R$ is an RHESB between $\mc C$ and $\mc D$.
If $R(X,Y)$ then $R(\mn X,\mn Y)$. 
\end{lemma}
\begin{proof}
Similar to that of Lemma~\ref{lem:rhsb min}.
Note that all events in $\mn X$ have depth one, and any transition involving only minimal elements is an equidepth one.
\end{proof}
\begin{lemma}\label{lem:rhesb lift}
Let $\mc C,\mc D \in \Cstable$.
Suppose that $R$ is an RHESB between $\mc C$ and $\mc D$.
Let $M \in C_{\mc C}$ be such that $\mn M = M$.
Similarly, let $N \in C_{\mc D}$ be such that $\mn {N} = N$.
Suppose also that $R(M,N)$.
Define $R_{M,N}$ by
\[
R_{M,N} =
\{(X \setminus M,Y \setminus N) :
R(X,Y),\ \mn {X} = M,\ \mn {Y} = N\}
\]
Then $R_{M,N}$ is an RHESB between $\mc C_M$ and $\mc D_N$.
\end{lemma}
\begin{proof}
Similar to that of Lemma~\ref{lem:rhsb lift}, using Lemma~\ref{lem:rhesb min} instead of Lemma~\ref{lem:rhsb min}.
\end{proof}

\begin{defn}
Let $\mc C \in \Cstable$.
For $m,n \in \Nat$ ($m \leq n$) and $X \in C_{\mc C}$, let
\[
\begin{array}{l}
X_{\leq n} \Defeq \{e \in X : \depth X e \leq n\}
\\
X_{\geq n} \Defeq \{e \in X : \depth X e \geq n\}
\\
X_{[m,n]} \Defeq \{e \in X : m \leq \depth X e \leq n\}
\end{array}
\]
\end{defn}
Clearly $X_{\leq n}$ is a configuration, since it is a left-closed subset of a configuration.
Also, $X_{\leq 1} = \mn X$.  For large enough $n$, $X_{\leq n} = X$.

\begin{prop}~\label{prop:rhesb levels}
Suppose that $R$ is an RHESB between $\mc C$ and $\mc D$.
If $R(X,Y)$ then for each $n \in \Nat$ we have $R(X_{\leq n},Y_{\leq n})$.
\end{prop}
\begin{proof}
Suppose that $R(X,Y)$.
We define configuration structures $\mc C_n,\mc D_n$ as follows:
\[
\begin{array}{l}
\mc C_0 \Defeq \mc C
\\
\mc C_{n+1} \Defeq (\mc C_n)_{X_{[n+1,n+1]}}
\end{array}
\]
(see Definition~\ref{def:cs lift}), and similarly for $\mc D_n$.
Note that $X_{[n+1,n+1]} = \mn{X_{\geq n+1}}$.
To ensure that $\mc C_n$ is well-defined, we show that $X_{\geq n+1}$ (and hence $X_{[n+1,n+1]}$) is a configuration of $\mc C_n$.  This is easily done by induction on $n$ (omitted).  Similarly for $\mc D_n$.

We also define RHESBs $R_n$ as follows:
\[
\begin{array}{l}
R_0 \Defeq R
\\
R_{n+1} \Defeq (R_n)_{X_{[n+1,n+1]},Y_{[n+1,n+1]}}
\end{array}
\]
(see Lemma~\ref{lem:rhesb lift}).
To ensure that $R_{n+1}$ is well-defined (and an RHESB) we need
$R_n(X_{[n+1,n+1]},Y_{[n+1,n+1]})$
to hold.
We can prove $R_n(X_{\geq n+1},Y_{\geq n+1})$ by an easy induction (omitted).
Using Lemma~\ref{lem:rhesb min}, we then deduce $R_n(X_{[n+1,n+1]},Y_{[n+1,n+1]})$.

Now we show $R(X_{\leq n},Y_{\leq n})$.  If $n = 0$ this is just $R(\emptyset,\emptyset)$, which is true by definition.  Suppose $n \geq 1$.  We show by induction on $i$ that for $1 \leq i \leq n$, $R_{n-i}(X_{[n-i+1,n]},Y_{[n-i+1,n]})$.
For $i = 1$ this is just $R_{n-1}(X_{[n,n]},Y_{[n,n]})$, which we have already shown.  Suppose that $R_{n-i}(X_{[n-i+1,n]},Y_{[n-i+1,n]})$ ($i<n$).
Then
\[
R_{n-i-1}(X_{[n-i+1,n]} \union X_{[n-i,n-i]},Y_{[n-i+1,n]} \union Y_{[n-i,n-i]})
\]
i.e.\ $R_{n-(i+1)}(X_{[n-(i+1)+1,n]},Y_{[n-(i+1)+1,n]})$, as required.

Putting $i = n$, we get $R_0(X_{[1,n]},Y_{[1,n]})$, i.e.\ $R(X_{\leq n},Y_{\leq n})$.
\end{proof}
\begin{cor}~\label{cor:rhesb levels}
Suppose that $R$ is an RHESB between $\mc C$ and $\mc D$.
If $R(X,Y)$ then for each $n \in \Nat$ we have $\lab(X_{[n,n]}) = \lab(Y_{[n,n]})$.
\end{cor}
\begin{proof}
By Proposition~\ref{prop:rhesb levels} and Lemma~\ref{lem:labels}.
\end{proof}
\begin{prop}~\label{prop:rhesb depth}
Let $\mc C,\mc D \in \Cstable$.
Suppose that $R$ is an RHESB between $\mc C$ and $\mc D$.
If $X \dtranc a k X'$, $Y \dtrand a {k'} Y'$, with $R(X,Y)$, $R(X',Y')$, then $k = k'$.
\end{prop}
\begin{proof}
By Corollary~\ref{cor:rhesb levels}, we know that $X,Y$ (resp.\ $X',Y'$) have the same multisets of events at each level $n$.  Hence the single events in $X' \setminus X$ and $Y' \setminus Y$ must have the same depth.
\end{proof}

So if $R(X,X')$ then
$X$ and $X'$ have similar structure, in that for each depth $n$, they have the same multisets of labelled events at that depth.  We can say that $X$ and $X'$ have similar amounts of concurrency (including auto-concurrency): $X$ and $X'$ have the same depth, and the same ``width'' at each level $n$.  Of course, this does not imply the stronger statement that $X$ and $X'$ are isomorphic, since we are not claiming that $X$ and $X'$ have the same causal relationships between the levels.
\begin{theorem}\label{thm:RSB=RDB}
On stable configuration structures,
${\rdbeqt} = {\rhesbeqt} = {\rsbeqt}$\ .
\end{theorem}
\begin{proof}
We have ${\rhesbeqt} \subseteq {\rdbeqt}$ by Proposition~\ref{prop:rhesb depth}.
Conversely, ${\rdbeqt} \subseteq {\rhesbeqt}$ by Proposition~\ref{prop:rdb-rhesb}.

We can prove versions of Proposition~\ref{prop:rhesb levels}, Corollary~\ref{cor:rhesb levels},
Proposition~\ref{prop:rhesb depth},
where we replace ``RHESB'' by ``RHSB''.
This shows that ${\rhsbeqt} \subseteq {\rdbeqt}$.
But ${\rsbeqt} \subseteq {\rhsbeqt}$, and so ${\rsbeqt} \subseteq {\rdbeqt}$.
Conversely, ${\rdbeqt} \subseteq {\rsbeqt}$ by Proposition~\ref{prop:rdb-rsb}.
\end{proof}
\begin{figure}
\centering
\psfrag{ib}{ib}
\psfrag{rib}{rb}
\psfrag{sb}{sb}
\psfrag{db}{db}
\psfrag{rsb=rhsb=rhesb=rdb}{rsb = rhsb = rhesb = rdb}
\psfrag{hh}{hh}
\centering
\includegraphics[width=1.2in]{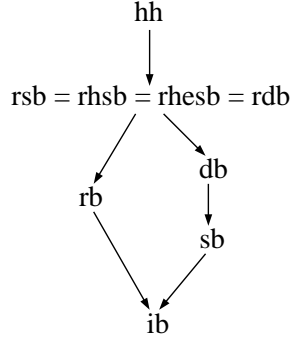}
\caption{Equivalences discussed in this paper.}\label{hierarchy1}
\end{figure}

\begin{rem}
The proof of Theorem~\ref{thm:RSB=RDB} shows that ${\rhsbeqt} \subseteq {\rdbeqt} \subseteq {\rsbeqt}$.
This gives us an alternative proof that any RHSB is an RSB (Theorem~\ref{thm:RHSB=RSB}),
bypassing Lemmas~\ref{lem:fwd step} and~\ref{lem:rev step}.
\end{rem}
Figure~\ref{hierarchy1} shows a diagram of the equivalences discussed in this paper;
arrows represent inclusion.

\section{Reverse Bisimulation and HH bisimulation}\label{sec:rb hh}

In this section we extend a result of Bednarczyk which shows that RB equivalence is as strong as HH equivalence in the absence of auto-concurrency.  We make use of results from Section~\ref{subsec:rdb} on depth-respecting bisimulations.

\begin{defn}
We say that $\mc C \in \Cstable$ is {\em without auto-concurrency} if for any $X \in C_{\mc C}$ and any $d,e \in X$, if $d \co_X e$ and $\lab(d) = \lab(e)$ then $d = e$.
\end{defn}

Bednarczyk showed the following:

\begin{theorem}[\cite{Bed91}]\label{thm:bed}
For prime event structures, in the absence of auto-concurrency, ${\rbeqt} = {\hheqt}$\ .
\end{theorem}

Before extending Theorem~\ref{thm:bed} in Theorem~\ref{thm:equidepth} below, we need a definition.  We say that a configuration structure is {\em without equidepth auto-concurrency}, if for any pair of distinct concurrent events, they cannot have both the same label and the same depth.  This is less restrictive than ``without auto-concurrency'', in that we allow auto-concurrent events, providing they are at different depths.
For example, the configuration structure $a \Par b.a$ exhibits auto-concurrency between the two events labelled $a$, but not equidepth auto-concurrency, since the two $a$ events are at depths one and two respectively.

\begin{defn}
We say that $\mc C \in \Cstable$ is {\em without equidepth auto-concurrency} if for any $X \in C_{\mc C}$ and any $d,e \in X$, if $d \co_X e$ with both $\lab(d) = \lab(e)$ and $\depth X d = \depth X e$,  then $d = e$.
\end{defn}

Notice that without equidepth auto-concurrency, an RB is an RHESB, since all reverse homogeneous equidepth steps must have just a single event: $X \rdtranc A = X'$ with $A$ homogeneous implies that $\card A = 1$.
Therefore we immediately have the following versions of
Proposition~\ref{prop:rhesb levels},
Corollary~\ref{cor:rhesb levels} and Proposition~\ref{prop:rhesb depth}:

\begin{prop}~\label{prop:rb levels}
Let $\mc C,\mc D \in \Cstable$ be without equidepth auto-concurrency.
Suppose that $R$ is an RB between $\mc C$ and $\mc D$.
If $R(X,Y)$ then for each $n \in \Nat$ we have $R(X_{\leq n},Y_{\leq n})$.
\qed
\end{prop}

\begin{cor}~\label{cor:rb levels}
Let $\mc C,\mc D \in \Cstable$ be without equidepth auto-concurrency.
Suppose that $R$ is an RB between $\mc C$ and $\mc D$.
If $R(X,Y)$ then for each $n \in \Nat$ we have $\lab(X_{[n,n]}) = \lab(Y_{[n,n]})$.
\qed
\end{cor}
Note that the no equidepth auto-concurrency condition means that $\lab(X_{[n,n]}), \lab(Y_{[n,n]})$
are sets rather than multisets, since events at the same depth must be concurrent.

\begin{prop}~\label{prop:rb depth}
Let $\mc C,\mc D \in \Cstable$ be without equidepth auto-concurrency.
Suppose that $R$ is an RB between $\mc C$ and $\mc D$.
If $X \dtranc a k X'$, $Y \dtrand a {k'} Y'$, with $R(X,Y)$, $R(X',Y')$, then $k = k'$.
\qed
\end{prop}
Proposition~\ref{prop:rb depth} would not necessarily hold in the presence of equidepth auto-concurrency.
For example, we have $a \Par a \rbeqt a.a$, but the two $a$ events are both at depth one in the case of $a \Par a$,
whereas they are at depths one and two respectively in the case of $a.a$.
\begin{theorem}\label{thm:equidepth}
In the absence of equidepth auto-concurrency, ${\rbeqt} = {\hheqt}$\ .
\end{theorem}
\begin{proof}[Proof (sketch)]
Let $\mc C,\mc D \in \Cstable$ be without equidepth auto-concurrency.
Suppose that $R$ is an RB between $\mc C$ and $\mc D$.
Suppose that $R(X,Y)$.
Define $f_{X,Y}: X \to Y$ by $f_{X,Y}(d) = e$ where $e$ is the unique $e \in Y$ such that
$\depth Y e = \depth X d$ and $\lab(e) = \lab(d)$.
Then $f_{X,Y}$ is well-defined and is a bijection by Corollary~\ref{cor:rb levels}.
Clearly $f$ preserves labels and depth.

We can show that $f_{X,Y}$ is order-preserving, using Proposition~\ref{prop:rb depth}.
Note that events in $X$ and $Y$ are determined uniquely by their depth and label.
We can also show the following,
again using Proposition~\ref{prop:rb depth}:
\begin{claim}
If $X \tranc a X'$, $Y \trand a Y'$ and $R(X,Y)$, $R(X',Y')$, then
$f_{X',Y'}\res X = f_{X,Y}$.
\end{claim}

Now define $R'(X,Y,f)$ iff $R(X,Y)$ and $f = f_{X,Y}$ (for $X \in C_{\mc C}, Y \in C_{\mc D}$).
We claim that $R'$ is an HH bisimulation between $\mc C$ and $\mc D$.
We have already seen that $f_{X,Y}$ is an isomorphism between $(X, <_X,\lab_{\mc C}\res X)$ and $(Y, <_Y,\lab_{\mc D}\res Y)$.
We can use the Claim to show the remaining properties needed to
verify that $R'$ is an HH bisimulation.
\end{proof}
We have strengthened Theorem~\ref{thm:bed} in two ways: we use stable configuration structures rather than prime event structures, and we weaken the assumption of no auto-concurrency to no equidepth auto-concurrency.

\section{Conclusions}

We have investigated bisimulation based on forward and reverse concurrent steps, so-called reverse step bisimulation (RSB) equivalence, in the stable configuration structure setting.
We showed that forward steps are unnecessary, while reverse steps can be restricted to ones where all events have the same label, and even the same causal depth.
We further showed that the same power can be got with single event transitions, provided that we match events on depth as well as label.  Thus we can claim that, in the reversible setting, the observational power of concurrent events is equivalent to that of single events together with depth.

We used these results to extend Bednarczyk's work and show that reverse bisimulation, where forward and reverse transitions are single events, coincides with hereditary
history-preserving (HH) bisimulation in the absence of equidepth auto-concurrency (multiple events with the same label and
at the same
depth).

Future work includes investigating the power of the reverse forms of pomset 
bisimulation, where transitions are pomsets, and weak history preserving 
bisimulation. We also intend to study the versions of forward bisimulations, 
such as, for example, interleaving bisimulation, pomset bisimulation and 
weak history preserving bisimulation, that satisfy additionally the 
hereditary property as expressed by the last condition of Definition~\ref{def:hh}.





\paragraph*{Acknowledgements}
We thank the anonymous referees for their helpful comments and suggestions.
The second author acknowledges partial support by EPSRC grant EP/G039550/1.

\bibliographystyle{eptcs}

\end{document}